\documentclass{llncs}
\usepackage{amsmath,amssymb}
\usepackage{cases}
\usepackage[utf8]{inputenc}
\usepackage{graphicx,color}
\usepackage[T1]{fontenc}
\usepackage[]{algpseudocode}
\usepackage{algorithm}
\usepackage{float}
\usepackage{enumitem} 

\newtheorem{thm}{Theorem}
\newtheorem{lem}[thm]{Lemma}

\newtheorem{cnj}[thm]{Conjecture}

\newtheorem{clm}[thm]{Claim}

\newcommand{\ifreplace}{\iffalse} 
\newcommand{\ifpart}{\iftrue} 

\newcommand{\mmax}[1]{\overline{#1}}

\title{An exact upper bound on the size of minimal clique covers}
\author{Ryan McIntyre \and Michael Soltys}
\institute{California State University Channel Islands\\
Dept.\ of Computer Science\\
One University Drive\\
Camarillo, CA 93012, USA\\
\email{\{ryan.mcintyre466@myci.csuci.edu,michael.soltys@csuci.edu\}}}

\date{\today}

\begin{document}
\maketitle

\begin{abstract}
Indeterminate strings have received considerable attention in the
recent past; see for example~\cite{Christodoulakis201534}
and~\cite{indeterminates-2017}. This attention is due to their
applicability in bioinformatics, and to the natural correspondence
with undirected graphs.  One aspect of this correspondence is the fact
that the minimal alphabet size of indeterminates representing any
given undirected graph corresponds to the size of the minimal clique
cover of this graph.  This paper solves a related problem proposed
in~\cite{indeterminates-2017}: compute $\Theta_n(m)$, which is the
size of the largest possible minimal clique cover (i.e., an exact
upper bound), and hence alphabet size of
the corresponding indeterminate, of any graph on $n$ vertices and $m$
edges. 
\end{abstract}

\section{Introduction}

Given an undirected graph $G=(V,E)$, we say that
$C=\{C_1,C_2,\ldots,C_k\}$ is a {\em clique cover} of $G$ of size $k$
if each $C_i$ is a set of vertices comprising a clique, and $\cup
C=V$, and furthermore, given any edge $(u,v)\in E$, there is a $C_i$
that contains both $u$ and $v$. We denote $\theta(G)$ as the size of
the smallest clique cover of $G$ (\cite{roberts-1983}). 

Let $\mathcal{G}_n(m)$ be the set of all undirected graphs on $n$
vertices and $m$ edges; of course, $0\le m\le{n\choose 2}$.  Let
$\Theta_n(m)$ be the largest possible $\theta(G)$ for
$G\in\mathcal{G}_n(m)$.  In \cite[Problem~11]{indeterminates-2017} the
authors pose the following problem: describe the function
$\Theta_n(m)$ for every given $n$, and they provide as an example a
graph for $\Theta_7(m)$, where $m$ ranges over $\{0,1,\ldots,21\}$,
$21={7\choose 2}$ (see~\cite[Fig.~3]{indeterminates-2017}). Given the
fact that for $n>7$ the number of graphs quickly becomes unwieldy, it
is desirable to compute $\Theta_n(m)$ analytically, as $\Theta_n(m)$
provides an exact bound on the alphabet size of indeterminates
obtained from a graph on $n$ vertices and $m$ edges.

We already know from~\cite{indeterminates-2017} that for each $n$, the
global maximum is reached at $m=\lfloor n^2/4\rfloor$. The reason 
for this is that $\lfloor n^2/4 \rfloor$ is the largest number of edges
that can fit in a graph on $n$ vertices without forcing any triangles;
note that such a graph is simply a complete bipartite graph. On the
other hand, if a graph has no triangles and no singletons, the only
possible clique cover for such a graph consists of all the edges (more
precisely, the cover consists of all pairs $\{u,v\}$ where $e=(u,v)$
is an edge in the graph).

\ifreplace
In the interest of space, some proofs have been omitted.
\fi

We aim to characterize $\Theta_n(m)$ in our primary result,
Theorem~\ref{thm:final}.  We also establish Algorithm~\ref{alg:ccb},
which computes $\Theta_n(m)$ in linear time.  Our motivation comes
from~\cite{indeterminates-2017}, where $\Theta_n(m)$ would be used as
an upper bound for the size of a minimal alphabet for an indeterminate
string based on the edge and vertex counts of said string's
corresponding undirected graph. The hope, of course, is that
exploration of the structural causes behind the upper bound of
$\theta(G)$ will help us to better understand the problem of finding a
minimal or near-minimal clique cover, which corresponds directly to
finding a small alphabet for an indeterminate string.  While the
motivation comes from string processing, our results are primarily in
extremal graph theory. We'll apply theorems provided by
Mantel~\cite{M1907} (Theorem~\ref{thm:mantel}) and
Lov\'asz~\cite{lovasz-1968} (Theorem~\ref{thm:lovasz}) to prove that
$\Theta_n(m)$ has many recursive properties, which we will then use to
characterize it.

\begin{thm}[Mantel]\label{thm:mantel}
If a graph on $n$ vertices contains no triangle, then it contains at
most $\lfloor n^2/4\rfloor$ edges.
\end{thm}

The expression $\lfloor n^2/4\rfloor$ will be used frequently
throughout the paper, and so we abbreviate it as $\mmax{n}:=\lfloor
n^2/4\rfloor$. In general, for any expression $\text{exp}$, we let
$\mmax{\text{exp}}=\lfloor\text{exp}^2/4\rfloor$.

\begin{thm}[Lov\'asz]\label{thm:lovasz}
Given $G\in\mathcal{G}_n(m)$, let $k$ be the number of missing edges 
(i.e. $k={n\choose2}-m$),
and let $t$ be the largest natural number such that $t^2-t\leq k$. Then
$\theta(G)\leq k+t$. Moreover, this bound is exact for $k=t^2$ or $k=t^2-t$.
\end{thm}


Of course, as Lov\`asz's bound relies solely on the number of missing
edges, it also relies on the assumption that the vertex and edge counts
are arbitrarily large; the bound is exact at the specified values of $k$, 
as long as that $m\geq\mmax{n}$.

Clearly, $t^2=\mmax{2t}$, and $t^2-t=\mmax{2t-1}$ (or identically 
$t^2+t=\mmax{2t+1}$), so Theorem~\ref{thm:lovasz} 
can be restated:

{\it Given $G\in\mathcal{G}_n(m)$, define $k$ as above
and let $t$ be the largest natural number such that 
$\mmax{2t-1}\leq k$. $\theta(G)\leq k+t$. Assuming $m\geq\mmax{n}$, this
bound is sharp if $k=\mmax{2t-1}$ or $k=\mmax{2t}$. That is, 
$\Theta_n(m)=\mmax{2t-1}+t=\mmax{2t}$ if 
$m={n\choose2}-\mmax{2t-1}\geq\mmax{n}$
and $\Theta_n(m)=\mmax{2t}+t=\mmax{2t+1}$ if 
$m={n\choose2}-\mmax{2t}\geq\mmax{n}$.}

We propose an improvement to
Theorem~\ref{thm:lovasz} in Conjecture~\ref{cnj:seventeenth}. In lieu of
proof for all $m$, we provide Lemmas~\ref{lem:sixteenth}, 
\ref{lem:eighteenth} and \ref{lem:nineteenth}, which prove that 
Conjecture~\ref{cnj:seventeenth} is true for some $m$.
The conjecture reads:

{\it Given $m\geq\mmax{n}$, define $k$ as above and let $t$ be the largest 
natural number such that $\mmax{t}\leq k$. $\Theta_n(m)=\mmax{t+1}$}

We use $i(G)$ to denote the number of singletons (or isolated
vertices) in $G$, and $c(G)$ to denote the number of non-isolated
vertices in $G$. Of course, for any graph $G$ on $n$ vertices we have
$i(G)+c(G)=n$. Let $I(G)$ denote the subgraph of $G$ consisting of the
isolated points in $G$ ($i(G)=|I(G)|$), and $C(G)$ denote the subgraph
of $G$ consisting of all of the non-isolated points in $G$
($c(G)=|C(G)|$).  Let $S(G)$ denote those vertices in $C(G)$ which are
connected by an edge to every other vertex in $C(G)$, i.e.,
$S(G)=\{v\in C(G): \forall u\in C(G)-\{v\},(v,u)\in E\}$. We call such
vertices {\em stars}. Let
$s(G)=|S(G)|$. Finally, $\hat{G}$ will denote the subgraph of $G$
which results from removing all vertices in $S(G)$, along with their
edges, but with one exception: if $C(G)$ is a clique, it is simply
replaced with a new singleton vertex $v_{C(G)}$ in $\hat{G}$.

As discussed in~\cite[\S2]{indeterminates-2017}, based on the results
of Mantel and Erd{\"o}s (\cite{M1907,erdos}), $\Theta_n(m)$ achieves
its global maximum at precisely $\mmax{n}$, and $\Theta_n(m)$ is
non-decreasing for $m\le\mmax{n}$ and non-increasing for
$m\ge\mmax{n}$, and $\Theta_n(\mmax{n})=\mmax{n}$. This fixed point
corresponds to the situation where the number of edges is as large as
possible without forcing any triangles, i.e., $m=\mmax{n}$, and it is
precisely at this point when the best cover can be forced to include all
$\mmax{n}$ of the edges. See Figure~\ref{fig:results-8-delta} for
$\Theta_8(m)$, where $\mmax{n}=\mmax{8}=\lfloor 8^2/4\rfloor=16$.

We describe $\Theta_n(m)$ in two sections:
Section~\ref{sec:pre-maximum} for $m\le\mmax{n}$, which relies primarily on
Theorem~\ref{thm:mantel}, and
Section~\ref{sec:post-maximum}, for $m\ge\mmax{n}$, based on 
Theorem~\ref{thm:lovasz}. 
Note that we assume
throughout that $n\ge 4$, thus when we say ``for all $n$,'' we mean
``for all $n\ge4$.''

\section{$\Theta_n(m)$ for $m\le\mmax{n}$}\label{sec:pre-maximum} 

We prove a sequence of auxiliary results that will help us
characterize the graph of $\Theta_n(m)$ for $m\le\mmax{n}$. The
forthcoming 
material is rather technical, but the reader will find it
easier to follow by keeping the graph in
Figure~\ref{fig:results-8-delta} in mind.

\begin{figure}[h]

\begingroup
  \makeatletter
  \providecommand\color[2][]{%
    \GenericError{(gnuplot) \space\space\space\@spaces}{%
      Package color not loaded in conjunction with
      terminal option `colourtext'%
    }{See the gnuplot documentation for explanation.%
    }{Either use 'blacktext' in gnuplot or load the package
      color.sty in LaTeX.}%
    \renewcommand\color[2][]{}%
  }%
  \providecommand\includegraphics[2][]{%
    \GenericError{(gnuplot) \space\space\space\@spaces}{%
      Package graphicx or graphics not loaded%
    }{See the gnuplot documentation for explanation.%
    }{The gnuplot epslatex terminal needs graphicx.sty or graphics.sty.}%
    \renewcommand\includegraphics[2][]{}%
  }%
  \providecommand\rotatebox[2]{#2}%
  \@ifundefined{ifGPcolor}{%
    \newif\ifGPcolor
    \GPcolorfalse
  }{}%
  \@ifundefined{ifGPblacktext}{%
    \newif\ifGPblacktext
    \GPblacktexttrue
  }{}%
  \let\gplgaddtomacro\g@addto@macro
  \gdef\gplbacktext{}%
  \gdef\gplfronttext{}%
  \makeatother
  \ifGPblacktext
    \def\colorrgb#1{}%
    \def\colorgray#1{}%
  \else
    \ifGPcolor
      \def\colorrgb#1{\color[rgb]{#1}}%
      \def\colorgray#1{\color[gray]{#1}}%
      \expandafter\def\csname LTw\endcsname{\color{white}}%
      \expandafter\def\csname LTb\endcsname{\color{black}}%
      \expandafter\def\csname LTa\endcsname{\color{black}}%
      \expandafter\def\csname LT0\endcsname{\color[rgb]{1,0,0}}%
      \expandafter\def\csname LT1\endcsname{\color[rgb]{0,1,0}}%
      \expandafter\def\csname LT2\endcsname{\color[rgb]{0,0,1}}%
      \expandafter\def\csname LT3\endcsname{\color[rgb]{1,0,1}}%
      \expandafter\def\csname LT4\endcsname{\color[rgb]{0,1,1}}%
      \expandafter\def\csname LT5\endcsname{\color[rgb]{1,1,0}}%
      \expandafter\def\csname LT6\endcsname{\color[rgb]{0,0,0}}%
      \expandafter\def\csname LT7\endcsname{\color[rgb]{1,0.3,0}}%
      \expandafter\def\csname LT8\endcsname{\color[rgb]{0.5,0.5,0.5}}%
    \else
      \def\colorrgb#1{\color{black}}%
      \def\colorgray#1{\color[gray]{#1}}%
      \expandafter\def\csname LTw\endcsname{\color{white}}%
      \expandafter\def\csname LTb\endcsname{\color{black}}%
      \expandafter\def\csname LTa\endcsname{\color{black}}%
      \expandafter\def\csname LT0\endcsname{\color{black}}%
      \expandafter\def\csname LT1\endcsname{\color{black}}%
      \expandafter\def\csname LT2\endcsname{\color{black}}%
      \expandafter\def\csname LT3\endcsname{\color{black}}%
      \expandafter\def\csname LT4\endcsname{\color{black}}%
      \expandafter\def\csname LT5\endcsname{\color{black}}%
      \expandafter\def\csname LT6\endcsname{\color{black}}%
      \expandafter\def\csname LT7\endcsname{\color{black}}%
      \expandafter\def\csname LT8\endcsname{\color{black}}%
    \fi
  \fi
    \setlength{\unitlength}{0.0500bp}%
    \ifx\gptboxheight\undefined%
      \newlength{\gptboxheight}%
      \newlength{\gptboxwidth}%
      \newsavebox{\gptboxtext}%
    \fi%
    \setlength{\fboxrule}{0.5pt}%
    \setlength{\fboxsep}{1pt}%
\begin{picture}(7200.00,5040.00)%
    \gplgaddtomacro\gplbacktext{%
      \csname LTb\endcsname%
      \put(682,704){\makebox(0,0)[r]{\strut{}$0$}}%
      \put(682,1137){\makebox(0,0)[r]{\strut{}$2$}}%
      \put(682,1570){\makebox(0,0)[r]{\strut{}$4$}}%
      \put(682,2002){\makebox(0,0)[r]{\strut{}$6$}}%
      \put(682,2435){\makebox(0,0)[r]{\strut{}$8$}}%
      \put(682,2868){\makebox(0,0)[r]{\strut{}$10$}}%
      \put(682,3301){\makebox(0,0)[r]{\strut{}$12$}}%
      \put(682,3733){\makebox(0,0)[r]{\strut{}$14$}}%
      \put(682,4166){\makebox(0,0)[r]{\strut{}$16$}}%
      \put(682,4599){\makebox(0,0)[r]{\strut{}$18$}}%
      \put(4236,484){\makebox(0,0){\strut{}$\mmax{n}=16$}}%
      \put(814,484){\makebox(0,0){\strut{}$0$}}%
      \put(1242,484){\makebox(0,0){\strut{}$2$}}%
      \put(1670,484){\makebox(0,0){\strut{}$4$}}%
      \put(2097,484){\makebox(0,0){\strut{}$6$}}%
      \put(2525,484){\makebox(0,0){\strut{}$8$}}%
      \put(2953,484){\makebox(0,0){\strut{}$10$}}%
      \put(3381,484){\makebox(0,0){\strut{}$12$}}%
      \put(3809,484){\makebox(0,0){\strut{}$14$}}%
      \put(4664,484){\makebox(0,0){\strut{}$18$}}%
      \put(5092,484){\makebox(0,0){\strut{}$20$}}%
      \put(5520,484){\makebox(0,0){\strut{}$22$}}%
      \put(5947,484){\makebox(0,0){\strut{}$24$}}%
      \put(6375,484){\makebox(0,0){\strut{}$26$}}%
      \put(6803,484){\makebox(0,0){\strut{}$28$}}%
      \put(1456,4819){\makebox(0,0){\strut{}$\delta_1$}}%
      \put(1883,4819){\makebox(0,0){\strut{}$\delta_1$}}%
      \put(2418,4819){\makebox(0,0){\strut{}$\delta_2$}}%
      \put(3060,4819){\makebox(0,0){\strut{}$\delta_2$}}%
      \put(3809,4819){\makebox(0,0){\strut{}$\delta_3$}}%
    }%
    \gplgaddtomacro\gplfronttext{%
      \csname LTb\endcsname%
      \put(176,2651){\rotatebox{-270}{\makebox(0,0){\strut{}size of largest minimal clique cover}}}%
      \put(3808,154){\makebox(0,0){\strut{}number of edges on 8 vertices}}%
    }%
    \gplbacktext
    \put(0,0){\includegraphics{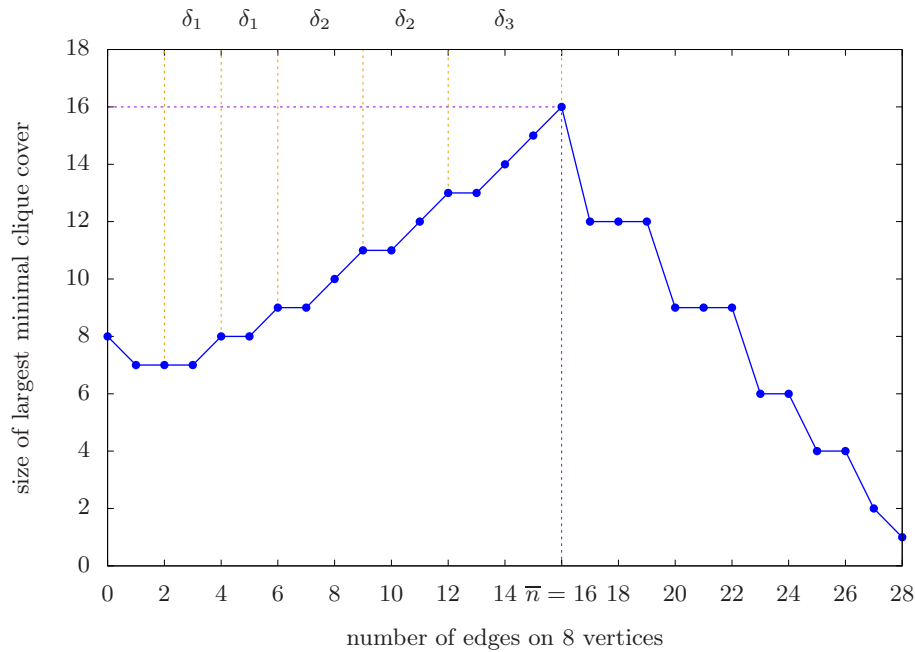}}%
    \gplfronttext
  \end{picture}%
\endgroup
\caption{Left side of the graph for $\Theta_8(m)$}\label{fig:results-8-delta}
\end{figure}

\begin{clm}\label{clm:third}
$\Theta_n(0)=\Theta_n(4)=n$ and
$\Theta_n(1)=\Theta_n(2)=\Theta_n(3)=n-1$.
\end{clm}

\ifpart
Claim~\ref{clm:third} is trivial, as it can be shown very
quickly through enumeration of all possible arrangements of 0, 1, 2, 3
and 4 edges.
\fi

\begin{clm}\label{clm:fourth}
$\Theta_n(m+1)\le\Theta_n(m)+1$.
\end{clm}

\ifpart
\begin{proof}
Consider $G\in\mathcal{G}_n(m+1)$.
Choose any edge $e$ in $G$, and remove it (while keeping its end-points)
to obtain $G_0\in\mathcal{G}_n(m)$. Let $C_0$ be the smallest clique cover of
$G_0$, and so $|C_0|\leq\Theta_n(m)$. Let
$C = C_0 \cup \{e\}$. Since $C_0$ covered all of $G$ except for $e$,
$C$ by extension covers $G$, and $|C|=|C_0|+1\leq\Theta_n(m)+1$. We have
found a cover for $G$ with cardinality of at most $\Theta_n(m)+1$.
\qed
\end{proof}
\fi

\begin{lem}\label{lem:fifth}
$m\leq\mmax{n}\wedge G\in\mathcal{G}_n(m)\wedge\theta(G)=\Theta_n(m)
\implies G$ is triangle-free.
\end{lem}

\begin{proof}
We will prove this Lemma by contrapositive. That is, we will show that if
$G\in\mathcal{G}_n(m)$, $m\leq\mmax{n}$, and $G$ contains
triangles, then $\theta(G)\neq\Theta_n(m)$.

Let $G\in\mathcal{G}_n(m)$, and assume $G$ has at least one triangle.
3 edges can be covered with 1 clique, so 
$\theta(G)\leq m-2+i(G)$.

{\bf Case 1:} $m>\mmax{c(G)}$.
We must first note that
$i(G)\neq0$, because if $i(G)=0$ then $c(G)=n$, so
$m>\mmax{n}$; this directly contradicts the assumptions of
this Lemma.

$\theta(G)\leq\mmax{c(G)}+i(G)$; because the
largest clique cover of any graph on $c(G)$ vertices is
$\mmax{c(G)}$, we need only add the singletons
to this bound to get an upper bound for $\theta(G)$.
Note that $c(G)\geq3$, since $G$ has a triangle. Consider a graph
$G_1\in\mathcal{G}_n(m)$ such that $c(G_1)=c(G)+1$ and
$i(G_1)=i(G)-1$. Such a graph can be constructed; we have more edges than can
fit on $c(G)$ vertices without triangles, so we can simply choose 1 edge
$\{u,v\}$ from a triangle in $G$, remove it, and replace it with $\{u,s\}$,
where $s$ is a singleton in $G$. We can write a similar bound for this graph:
$\theta(G_1)\leq\mmax{c(G_1)}+i(G_1)$. Let's
compare the two bounds. 

If $c(G)$ is even, then
$\mmax{(c(G)+1)}-\mmax{c(G)}=
{c(G)/2}$. Since $c(G)$ is even and $\geq3$, $c(G)\geq4$. So
$\mmax{(c(G)+1)}-\mmax{c(G)}\geq2$.

If $c(G)$ is odd,
$\mmax{(c(G)+1)}-\mmax{c(G)}=
{(c(G)+1)}/{2}$. Given that $c(G)\geq3$, this difference is,
again, $\geq2$. 

As a result, our upper bound for $\theta(G_1)$ is at least 1 more than that for 
$\theta(G)$. Of course, this isn't enough to prove that 
$\theta(G_1)\geq\theta(G)$. We can, however, repeat this process to get 
increasing upper bounds for $\theta(G_2)$ on $c(G_1)+1$ non-isolated vertices 
and $i(G_1)-1$ singletons, and so on until we reach a $G_\alpha$ such that 
$\mmax{c(G_\alpha)}\geq m$. Note that since $m\leq\mmax{n}$, this will
necessarily happen before or when we run out of singletons. 

If $\alpha=1$, then $\mmax{c(G)+1}\geq m$, so we can construct
triangle-free $H\in\mathcal{G}_{c(G)+1}(m)$. Let $I$ be a graph of $i(G)-1$
singletons, and let $G'=H\cup I$. $G'\in\mathcal{G}_n(m)$, and
$\theta(G')=m-1+i(G)$. Recall that $\theta(G)\leq m-2+i(G)$, so
$\theta(G')>\theta(G)$.

If $\alpha>1$, then we can
construct triangle-free $H\in\mathcal{G}_{c(G_\alpha)}(m)$. Let $I$ be the graph
of $i(G_\alpha)$ singletons, and let $G'=H\cup I$. $G'\in\mathcal{G}_n(m)$ and
$\theta(G')=m+i(G')$ by construction. Let $B_\beta$ 
denote the previously established upper bound for $\theta(G_\beta)$ 
($1\leq\beta\leq\alpha$), and let $B$ denote the upper bound established for 
$G$. Note that $m\geq\mmax{c(G_{\alpha-1})}+1$ or we would have stopped prior to 
$G_\alpha$, and that $i(G')=i(G_{\alpha-1})-1$. Thus, $\theta(G')\geq
B_{\alpha-1}>B\geq\theta(G)$, so $\theta(G')>\theta(G)$. 

Regardless of $\alpha$'s value, we have found $G'\in\mathcal{G}_n(m)$ such that 
$\theta(G')>\theta(G)$, so $\Theta_n(m)\neq\theta(G)$.

{\bf Case 2:} $m\leq\mmax{c(G)}$.
We can
construct a triangle-free graph $H\in\mathcal{G}_{c(G)}(m)$. Since
$H$ has no triangles, $\theta(H)\geq m$.  Let $I$ be the graph of
$i(G)$ singletons, and let $G'=H\cup I$. Then $G'\in\mathcal{G}_n(m)$
and $\theta(G')\geq m+i(G)>m+i(G)-2\geq\theta(G)$. Thus,
$\theta(G)\neq\Theta_n(m)$.
 
In either case, we have shown that if $m\leq\mmax{n}$, $G\in\mathcal{G}_n(m)$, and $G$ contains 
at least 1 triangle, then $\theta(G)\neq\Theta_n(m)$. Thus, if $m\leq\mmax{n}$
and $G\in\mathcal{G}_n(m)$, then $\theta(G)=\Theta_n(m)$ implies that $G$ is
triangle-free.
\qed
\end{proof}


\begin{lem}\label{lem:sixth}
If $m\leq\mmax{n}$ and $m=p^2$ or $m=p(p+1)$ for some positive 
integer $p$, then $\Theta_n(m)=m+n-2p$ or $\Theta_n(m)=m+n-2p-1$, respectively.
\end{lem}

\begin{proof} 
{\bf Case 1:} $m=p^2$.
Consider the complete bipartite graph
$K_{p,p}$. Since it has no triangles or singletons, 
$\theta(K_{p,p})=m$. Let $I$ be a
graph consisting of $n-2p$ singleton vertices, and let $G = K_{p,p}
\cup I$. $K_{p,p}$ has $2p$ vertices, and $I$ has $n-2p$ vertices, so
$G$ has $n$ vertices.  Similarly, $K_{p,p}$ has $m$ edges and $I$ has
$0$ edges, so $G$ has $m$ edges.  Thus, $G\in\mathcal{G}_n(m)$.
Moreover, since $G$ is triangle-free, $\theta (G)=m+n-2p$; that is,
$G$'s smallest clique cover is equal to its edge count plus its number
of singleton vertices. 

Let $H\in\mathcal{G}_n(m)$ such that $H$ is triangle-free, and let
$H_0$ be $H$ without its singletons. Then $H_0$ is also triangle-free,
and $H_0\in\mathcal{G}_{c(H)}(m)$. Mantel's Theorem shows that $m \leq
\frac{c(H)^2}{4}$, or identically that $c(H) \geq 2\sqrt{m}$. Since
$m=p^2$, this in turn is equivalent to $c(H)\geq 2p$. $H$ is not
necessarily bipartite, but it can be covered by one clique for each
edge (i.e., a clique consisting of the edge's incident vertices) plus
one clique for each singleton vertex. That is, there is a cover of $H$
with cardinality $m+n-c(H)\leq m+n-2p$.  Thus, $\theta(H)\leq m+n-2p$.
In conclusion, if $m=p^2$ and $m \leq \mmax{n}$ then
$\Theta_n(m)=m+n-2p$.

{\bf Case 2:} $m=p(p+1)$.
Consider the complete bipartite graph
$K_{p,p+1}$. Again, since it's bipartite, $\theta(K_{p,p+1})=m$.
$K_{p,p+1}$ has $2p+1$ vertices, so let $I$ be the graph consisting of
$n-2p-1$ singleton vertices, and let $G=K_ {p,p+1}\cup I$.
$G\in\mathcal{G}_n(m)$, and $\theta(G)=m+n-2p-1$. Similar to the
previous case, given triangle-free $H\in\mathcal{G}_n(m)$, we can
bound the number of non-isolated vertices in $H$:
$c(H)\geq2\sqrt{p^2+p}$. Since $c(H)$ is an integer, this bound can be
improved to $c(H)\geq\lceil2\sqrt{p^2+p}\rceil$. Obviously,
$2\sqrt{p^2}<2\sqrt{p^2+p}<2\sqrt{p^2+p+0.25}$. Identically,
$2p<2\sqrt {p^2+p}<2p+1$. So $\lceil2\sqrt{p^2+p}\rceil=2p+1$. Thus,
$c(H)\geq 2p+1$. Again, since $H$ can be covered
by its individual edges and singletons, $\theta(H)\leq
m+n-2p-1$. We have shown that if $m=p(p+1)$ and $m\leq\mmax{n}$ then
$\Theta_n(m)=m+n-2p-1$.
\qed
\end{proof}

\begin{lem}\label{lem:seventh}
If $m<\mmax{n}$ and $m=p^2$ or $m=p(p+1)$, then
$\Theta_n(m)=\Theta_n(m+1)$.
\end{lem}

\begin{proof}
{\bf Case 1:} $m=p^2$. Let $G\in\mathcal{G}_n(m+1)$ be triangle-free. Mantel's 
Theorem grants $c(G)\geq2\sqrt{p^2+1}$. Since $c(G)$ is an integer,
$c(G)\geq\lceil2\sqrt{p^2+1}\rceil$. Obviously $2\sqrt{p^2+1}>2p$, so
$c(G)\geq2p+1$; $G$ has $m+1$ edges and at most $n-2p-1$ singletons.
Thus, $\theta(G)\leq m+1+n-2p-1=m+n-2p$; in other words,
$\theta(G)\leq\Theta_n(m)$. A graph $G\in\mathcal{G}_n(m+1)$ for which 
$\theta(G)=\Theta_n(m)$ can be constructed easily. For example, let $I$ be the 
graph of $n-2p$ singletons, and let $G=K_{p,p}\cup I\cup\{e\}$, where $e$ 
is an edge with one incident vertex in $K_{p,p}$ and the other in $I$.

{\bf Case 2:} $m=p(p+1)$. Let $G\in\mathcal{G}_n(m+1)$ be triangle-free. 
Mantel's Theorem, combined with the fact that $c(G)$ is an integer, grants 
$c(G)\geq\lceil2\sqrt{p^2+p+1}\rceil$. 
$2\sqrt{p^2+p+0.25}<2\sqrt{p^2+p+1}<2\sqrt{p^2+2p+1}$, so
$2p+1<2\sqrt{p^2+p+1}<2p+2$. Therefore, $c(G)\geq2p+2$. $G$ is a graph with
$m+1$ edges and at most $n-2p-2$ singletons. As such, $\theta(G)\leq
m+n-2p-1$; $\theta(G)\leq\Theta_n(m)$. A graph $G\in\mathcal{G}_n(m+1)$ for 
which $\theta(G)=\Theta_n(m)$ can be constructed in much the same way as above.

In both cases, we have shown that $\Theta_n(m+1)=\Theta_n(m)$.
\qed
\end{proof}

\begin{lem}\label{lem:eighth}
$(\forall m\leq\mmax{n})$ $\Theta_{n+1}(m)=\Theta_n(m)+1$.
\end{lem}

\begin{proof}

Let $G_0\in\mathcal{G}_n(m)$ such that $\theta(G_0) = \Theta_n(m)$.
Lemma~\ref{lem:fifth} ensures that $G_0$ is triangle-free.
Therefore, $\theta(G_0)=m+i(G_0)$; $G_0$ could be any triangle-free graph on 
$n$ vertices and $m$ edges which maximizes $i(G)$. Maximizing $i(G_0)$ is 
identical to minimizing $c(G_0)$, as $i(G_0)=n-c(G_0)$. So let $n'$ be the
smallest number of vertices which can contain $m$ edges without a triangle.
$\theta(G_0)=m+n-n'$.
Similarly, let $G_1\in\mathcal{G}_{n+1}(m)$ be a triangle-free graph such that
$\theta(G_1)=\Theta_{n+1}(m)$. Again, $c(G_1)=n'$, as the number of edges is the
same, so $\Theta_{n+1}(m)=\theta(G_1)=m+(n+1)-n'=\theta(G_0)+1=\Theta_n(m)+1$.
\qed

\end{proof}

We are now ready to describe $\Theta_n(m)$ for $m\le\mmax{n}$. From
Lemma~\ref{lem:eighth}, we can first take the entirety of $\Theta$
for $n-1$ vertices (up to its maximum) and add 1 to every dependent value.
This establishes the portion of $\Theta_n$ ranging from $0$ to
$\mmax{(n-1)}$ edges, and
gives us a current right-most point at
$(\mmax{(n-1)},\mmax{(n-1)}+1)$. Then, Lemma~\ref{lem:seventh} grants
that the next point is $(\mmax{(n-1)}+1,\mmax{(n-1)}+1)$. From here,
the plot must make it to $(\mmax{n},\mmax{n})$; the increase in cover
size is equal to the increase in edge count. In other words, from this
point on the cover size must increase by 1 for each edge, on average.
This, combined with Claim~\ref{clm:fourth}, shows that it actually
must increase by exactly 1 per additional edge up to
$(\mmax{n},\mmax{n})$.
 
So, essentially, to get the left side of $\Theta_n$, simply take the
left side of $\Theta_{n-1}$, shift it upward by 1, then add the 
portion of the line $m=n$ ranging from $(\mmax{(n-1)}+1, \mmax{(n-1)}+1)$ 
up to the new maximum at $(\mmax{n},\mmax{n})$. We need
only find the horizontal length of this added segment,
$d=\mmax{n}-(\mmax{(n-1)}+1)$, to determine the pattern.

If $n$ is even, then $n-1$ is odd. So
$\mmax{n}=\frac{n^2}{4}$, and $\mmax{(n-1)}=\frac{(n-1)^2-1}{4}$.  So
$\mmax{n}-(\mmax{(n-1)}+1)=\frac{n^2}{4}-(\frac{(n-1)^2-1}{4}+1)$.
Simplification grants $d=\frac{n}{2}-1=\lfloor\frac{n}{2}\rfloor-1$. 

Similarly, if $n$ is odd then $n-1$ is even. As such,
$\mmax{n}=\frac{n^2-1}{4}$ and $\mmax{(n-1)}=\frac{(n-1)^2}{4}$.
Again, subtract and simplify to get
$d=\frac{n-1}{2}-1=\lfloor\frac{n}{2}\rfloor-1$. Note that this is the
same difference that was attained from the even number of vertices
directly preceding this odd $n$, and that in both cases, increasing
$n$ by two increases $d$ by 1.

Let
$\delta_d$ denote the sequence of pairs $(\Delta m,\Delta \Theta_n)$, comprised of 
the pair $(+1,+0)$ followed by $d$ pairs $(+1,+1)$. For
example, $\delta_2$ would be $\{(+1,+0),(+1,+1),(+1,+1)\}$. Recall
that the left-most points of $\Theta_n(m)$ are $(0,n)$, $(1,n-1)$, $(2,n-1)$,
$(3,n-1)$, $(4,n)$. This is the entire left side of $\Theta_4$.
To extend this to show the left side on 5 vertices, we need
only add $\delta_{\lfloor\frac{5-1}{2}\rfloor-1}$ (or $\delta_1$) to
the right of these points. $\delta_1 = \{(+1,+0),(+1,+1)\}$, so we get
2 additional points. Our last point was $(4,n)$. Addition of $(+1,+0)$
grants our first new point, $(5,n)$. Then, addition of $(+1,+1)$
grants $(6,n+1)$. To extend this to be the left side of the graph on 6
vertices, we would then add $\delta_2$ since
$\lfloor\frac{6}{2}\rfloor-1=2$. Then $\delta_2$ to get to 7 vertices,
then $\delta_3$ twice to get to 8 and 9 vertices, $\delta_4$ twice to
get 10 and 11 vertices, and so on.

The pattern is clearest if we start from the point $(2,n-1)$,
after which we have the sequence of changes:
$\delta_1,\delta_1,\delta_2,\delta_2,\delta_3,\delta_3,
\ldots,\delta_p,\delta_p,\ldots$,
until the maximum $(\mmax{n},\mmax{n})$ is reached.

\section{$\Theta_n(m)$ for $m\ge\mmax{n}$}\label{sec:post-maximum}

As in the previous section, the material is technical, but the reader
will find it easier to follows by keeping in mind
Figure~\ref{fig:results-8-gamma}.

\begin{figure}[h]

\begingroup
  \makeatletter
  \providecommand\color[2][]{%
    \GenericError{(gnuplot) \space\space\space\@spaces}{%
      Package color not loaded in conjunction with
      terminal option `colourtext'%
    }{See the gnuplot documentation for explanation.%
    }{Either use 'blacktext' in gnuplot or load the package
      color.sty in LaTeX.}%
    \renewcommand\color[2][]{}%
  }%
  \providecommand\includegraphics[2][]{%
    \GenericError{(gnuplot) \space\space\space\@spaces}{%
      Package graphicx or graphics not loaded%
    }{See the gnuplot documentation for explanation.%
    }{The gnuplot epslatex terminal needs graphicx.sty or graphics.sty.}%
    \renewcommand\includegraphics[2][]{}%
  }%
  \providecommand\rotatebox[2]{#2}%
  \@ifundefined{ifGPcolor}{%
    \newif\ifGPcolor
    \GPcolorfalse
  }{}%
  \@ifundefined{ifGPblacktext}{%
    \newif\ifGPblacktext
    \GPblacktexttrue
  }{}%
  \let\gplgaddtomacro\g@addto@macro
  \gdef\gplbacktext{}%
  \gdef\gplfronttext{}%
  \makeatother
  \ifGPblacktext
    \def\colorrgb#1{}%
    \def\colorgray#1{}%
  \else
    \ifGPcolor
      \def\colorrgb#1{\color[rgb]{#1}}%
      \def\colorgray#1{\color[gray]{#1}}%
      \expandafter\def\csname LTw\endcsname{\color{white}}%
      \expandafter\def\csname LTb\endcsname{\color{black}}%
      \expandafter\def\csname LTa\endcsname{\color{black}}%
      \expandafter\def\csname LT0\endcsname{\color[rgb]{1,0,0}}%
      \expandafter\def\csname LT1\endcsname{\color[rgb]{0,1,0}}%
      \expandafter\def\csname LT2\endcsname{\color[rgb]{0,0,1}}%
      \expandafter\def\csname LT3\endcsname{\color[rgb]{1,0,1}}%
      \expandafter\def\csname LT4\endcsname{\color[rgb]{0,1,1}}%
      \expandafter\def\csname LT5\endcsname{\color[rgb]{1,1,0}}%
      \expandafter\def\csname LT6\endcsname{\color[rgb]{0,0,0}}%
      \expandafter\def\csname LT7\endcsname{\color[rgb]{1,0.3,0}}%
      \expandafter\def\csname LT8\endcsname{\color[rgb]{0.5,0.5,0.5}}%
    \else
      \def\colorrgb#1{\color{black}}%
      \def\colorgray#1{\color[gray]{#1}}%
      \expandafter\def\csname LTw\endcsname{\color{white}}%
      \expandafter\def\csname LTb\endcsname{\color{black}}%
      \expandafter\def\csname LTa\endcsname{\color{black}}%
      \expandafter\def\csname LT0\endcsname{\color{black}}%
      \expandafter\def\csname LT1\endcsname{\color{black}}%
      \expandafter\def\csname LT2\endcsname{\color{black}}%
      \expandafter\def\csname LT3\endcsname{\color{black}}%
      \expandafter\def\csname LT4\endcsname{\color{black}}%
      \expandafter\def\csname LT5\endcsname{\color{black}}%
      \expandafter\def\csname LT6\endcsname{\color{black}}%
      \expandafter\def\csname LT7\endcsname{\color{black}}%
      \expandafter\def\csname LT8\endcsname{\color{black}}%
    \fi
  \fi
    \setlength{\unitlength}{0.0500bp}%
    \ifx\gptboxheight\undefined%
      \newlength{\gptboxheight}%
      \newlength{\gptboxwidth}%
      \newsavebox{\gptboxtext}%
    \fi%
    \setlength{\fboxrule}{0.5pt}%
    \setlength{\fboxsep}{1pt}%
\begin{picture}(7200.00,5040.00)%
    \gplgaddtomacro\gplbacktext{%
      \csname LTb\endcsname%
      \put(682,704){\makebox(0,0)[r]{\strut{}$0$}}%
      \put(682,1137){\makebox(0,0)[r]{\strut{}$2$}}%
      \put(682,1570){\makebox(0,0)[r]{\strut{}$4$}}%
      \put(682,2002){\makebox(0,0)[r]{\strut{}$6$}}%
      \put(682,2435){\makebox(0,0)[r]{\strut{}$8$}}%
      \put(682,2868){\makebox(0,0)[r]{\strut{}$10$}}%
      \put(682,3301){\makebox(0,0)[r]{\strut{}$12$}}%
      \put(682,3733){\makebox(0,0)[r]{\strut{}$14$}}%
      \put(682,4166){\makebox(0,0)[r]{\strut{}$16$}}%
      \put(682,4599){\makebox(0,0)[r]{\strut{}$18$}}%
      \put(4236,484){\makebox(0,0){\strut{}$\mmax{n}=16$}}%
      \put(814,484){\makebox(0,0){\strut{}$0$}}%
      \put(1242,484){\makebox(0,0){\strut{}$2$}}%
      \put(1670,484){\makebox(0,0){\strut{}$4$}}%
      \put(2097,484){\makebox(0,0){\strut{}$6$}}%
      \put(2525,484){\makebox(0,0){\strut{}$8$}}%
      \put(2953,484){\makebox(0,0){\strut{}$10$}}%
      \put(3381,484){\makebox(0,0){\strut{}$12$}}%
      \put(3809,484){\makebox(0,0){\strut{}$14$}}%
      \put(4664,484){\makebox(0,0){\strut{}$18$}}%
      \put(5092,484){\makebox(0,0){\strut{}$20$}}%
      \put(5520,484){\makebox(0,0){\strut{}$22$}}%
      \put(5947,484){\makebox(0,0){\strut{}$24$}}%
      \put(6375,484){\makebox(0,0){\strut{}$26$}}%
      \put(6803,484){\makebox(0,0){\strut{}$28$}}%
      \put(4771,4819){\makebox(0,0){\strut{}$\gamma_3$}}%
      \put(5413,4819){\makebox(0,0){\strut{}$\gamma_3$}}%
      \put(5947,4819){\makebox(0,0){\strut{}$\gamma_2$}}%
      \put(6375,4819){\makebox(0,0){\strut{}$\gamma_2$}}%
    }%
    \gplgaddtomacro\gplfronttext{%
      \csname LTb\endcsname%
      \put(176,2651){\rotatebox{-270}{\makebox(0,0){\strut{}size of largest minimal clique cover}}}%
      \put(3808,154){\makebox(0,0){\strut{}number of edges on 8 vertices}}%
    }%
    \gplbacktext
    \put(0,0){\includegraphics{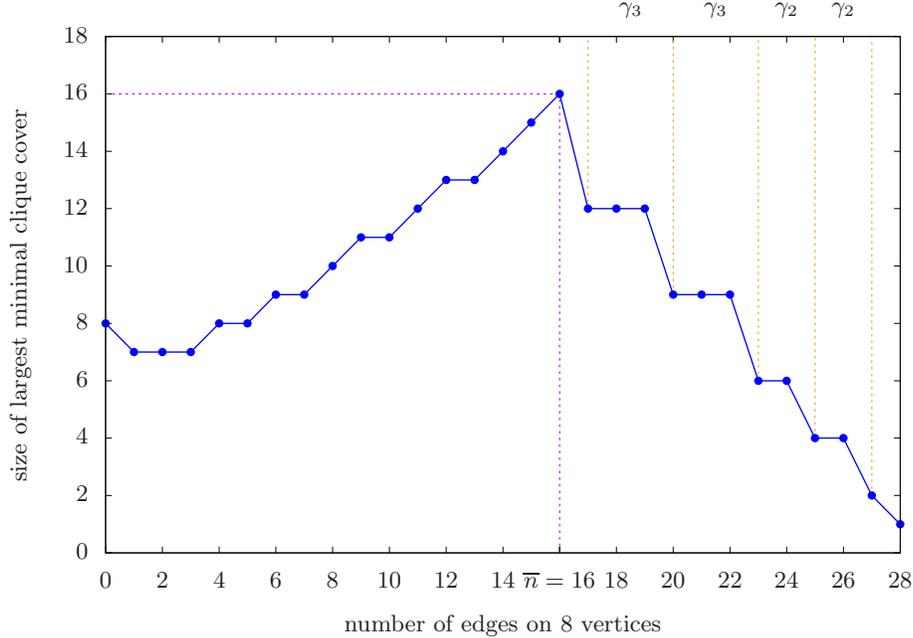}}%
    \gplfronttext
  \end{picture}%
\endgroup
\caption{Right side of the graph for
$\Theta_8(m)$}\label{fig:results-8-gamma}
\end{figure}



In the style of the $\delta_d$'s discussed at the end of the
previous section, let $\gamma_d$ denote the change in coordinates
$(-1,+d)$ followed by $d-1$ iterations of the change
$(-1,+0)$. For example, $\gamma_3$ is $\{(-1,+3),(-1,+0),(-1,+0)\}$.
If we start from $({n\choose 2}+1,0)$ and move left, we can construct
the right side of the graph with the sequence of changes
$\gamma_1,\gamma_1,\gamma_2,\gamma_2,\gamma_3,\gamma_3,
\ldots,\gamma_p,\gamma_p,\ldots$ until the maximum at
$(\mmax{n},\mmax{n})$ is reached (with the
``$(-1,+0)$''s trailing the last $\gamma$ omitted).

Recall the definitions of $I(G),i(G),C(G),c(G),S(G),s(G)$ from the
introduction: isolated vertices and their count, non-isolated vertices
and their count, and stars and their count, respectively.  Also recall
that if $C(G)$ is not a clique, then $\hat{G}$ denotes the subgraph of
$G$ which results from removing all stars, along with their edges,
from $G$. If $C(G)$ is a clique, then $\hat{G}$ is the graph which
results from replacing $C(G)$ with a single new vertex. 

Additionally, in order to keep the notation as simple as possible, 
we will assume that the intersection or union of a vertex set and 
a graph includes edges whenever convenient.

We prove below that removing $S(G)$ does not change $\theta(G)$,
but an identical proof works for any subset of $S(G)$.

\begin{lem}\label{lem:ninth}
\begin{enumerate}[label=(\roman*)]
	\item $\theta(\hat{G})=\theta(G)$. 
	\item The removal of any subset of $S(G)$ does not decrease $\theta$.
	\item Due to the above, if $G\in\mathcal{G}_{n+1}(m+n)$ and 
		$s(G)\neq0$ then $\theta(G)\leq\Theta_{n}(m)$
\end{enumerate}
\end{lem}

\begin{proof}
First, note that $\theta(G)=\theta(C(G))+i(G)$, and similarly that
every vertex in $I(G)$ is also a singleton in $\hat{G}$, so
$\theta(\hat{G})=\theta(C(G)\cap\hat{G})+i(G)$ if $C(G)$ is not a
clique, and $1+i(G)$ otherwise. As such, we can assume
without loss of generality that $G$ has no singletons (i.e.,
$C(G)=G$).

If $C(G)$ is complete, then it can be covered by 1 clique, so
replacing it with a singleton vertex has no effect on $\theta$;
$\theta(\hat{G})=\theta(G)$.

So let $G$ be a graph with no singletons such that $C(G)$ is not
complete, and let $C$ be a minimal clique cover of $G$ consisting
entirely of maximal cliques. We know such a cover exists
from~\cite{indeterminates-2017}. Let $C_0=\bigcup_{c\in
C}\{c\cap\hat{G}\}$. The elements of $C_0$ are still cliques in
$\hat{G}$, as any pair of vertices in $\hat{G}$ which were connected
in $G$ are still connected in $\hat{G}$. Moreover, $C_0$ covers
$\hat{G}$; every edge and vertex of $\hat{G}$ was covered by $C$, and
the only difference from $C$ to $C_0$ is the removal of the vertices
and edges which were not included in $\hat{G}$. It remains to be seen
that the elements of $C_0$ are nonempty and unique. Let $c_0$ be an
element of $C_0$. Then there is a maximal $c\in C$ such that
$c_0=c\cap\hat{G}=c-S(G)$. Moreover, since $C(G)$ is incomplete, there
is a vertex in $C(G)$ which is not in $S(G)$, and since $S(G)$ is
fully connected to any element of $C(G)$, $S(G)$ cannot be a maximal
clique. Thus every $c\in C$ contains a vertex not in $S(G)$, so every
$c_0\in C_0$ is nonempty. Moreover, $G$ has no singletons so every
$c\in C$ is a subset of $C(G)$, and as such $(\forall c\in C)(c\supset
S(G))$. So each $c_0\in C_0$ is the result of removing the entirety
of $S(G)$, from a clique $c\in C$.
Therefore, if two elements of $C_0$ are identical, then their
corresponding elements from $C$ were identical, so $C$ is not a
smallest clique cover, as a duplicate clique could be removed to find
a smaller one. Thus, every element of $C_0$ is nonempty and unique. We
have found a cover $C_0$ of $\hat{G}$ such that $|C_0|=|C|$. 

Assume that there is a cover $C_0'$ of $\hat{G}$, composed of maximal
cliques, such that $|C_0'|<|C_0|$. Let $C'=\bigcup_{c\in C_0'}\{c\cup
S(G)\}$. Obviously, $C'$ covers any edge or vertex in $\hat{G}$. So
let $v$ be a vertex in $G$ that is not in $\hat{G}$. $v\in S(G)$,
so $v$ is covered by every clique in $C'$. Let $e$ be an edge in $G$
but not in $\hat{G}$. Then $e$ either has two incident vertices in
$S(G)$ or one in $S(G)$ and the other in $\hat{G}$. If both incident
vertices are in $S(G)$, then $e$ is covered by every clique in $C'$,
as $(\forall c\in C')(c\supset S(G))$. Alternately, if one incident
vertex is in $\hat{G}$, then there is a cover $c_0\in C_0'$ which
contains this vertex; thus there is a $c'\in C'$ such that $c'=c_0\cup
S(G)$ which covers $e$. So $C'$ covers $G$, and
$|C'|=|C_0'|<|C_0|=|C|$; $C$ is not a smallest cover of $G$, which
directly contradicts its definition.

We have found a cover for $\hat{G}$ with the same cardinality as 
a minimal cover of $G$, and shown that a smaller cover for $\hat{G}$ 
cannot exist. Thus, $\theta(\hat{G})=\theta(G)$. 
\qed
\end{proof}

\begin{clm}\label{clm:tenth}
If $G\in\mathcal{G}_n(m)$ and
$\hat{G}\in\mathcal{G}_{\hat{n}}(\hat{m})$, then $\hat{n}=n-s(G)$, and 
$$
\hat{m}=m-{s(G)\choose2}-s(G)(c(G)-s(G))=m-\frac{s(G)(2c(G)-s(G)-1)}{2}.
$$
\end{clm}

\ifpart
\begin{proof}
$\hat{n}=n-s(G)$ follows directly from the definitions of $\hat{G}$
and $s(G)$.

In order to get $\hat{G}$ from $G$, we remove every vertex in $S(G)$
and their incident edges. There are two types of edges which are
removed: those connecting vertices in $S(G)$ to each other, and those
connecting vertices in $S(G)$ to vertices not in $S(G)$. $S(G)$ is a
complete subgraph of $G$, so its removal results in the removal of
${s(G)\choose2}$ edges. Every vertex in $S(G)$ is also fully connected
to the remaining vertices in $C(G)$, of which there are $c(G)-s(G)$.
Thus, removing these edges lowers the edge count by $s(G)(c(G)-s(G))$.
This grants $\hat{m}=m-{s(G)\choose2}-s(G)(c(G)-s(G))$.
Simplification grants $\hat{m}=m-\frac{s(G)(2c(G)-s(G)-1)}{2}$.
\qed
\end{proof}
\fi

\begin{clm}\label{clm:eleventh}
$c(\hat{G})\leq c(G)-s(G)$, and $i(\hat{G})\geq i(G)$
\end{clm}

\ifpart
\begin{proof}
From the definitions of $S(G)$ and $\hat{G}$, it is clear that
$C(\hat{G})\subseteq C(G)-S(G)$, so $c(\hat{G})\leq c(G)-s(G)$.
Similarly, $I(\hat{G})\supseteq I(G)$, so $i(\hat{G})\geq i(G)$. Note
that both inequalities result from fact that $\hat{G}$ may have
singletons which were not isolated in $G$; specifically, if a vertex
$v\in C(G)$ only has edges incident to vertices in $S(G)$, then $v\in
I(\hat{G})$.
\qed
\end{proof}
\fi

\begin{clm}\label{clm:twelfth}
If $C(G)$ is not complete, then $\theta(G)\leq\mmax{c(G)-s(G)}+i(G)$.
If $C(G)$ is complete, then $\theta(G)=1+i(G)$. 
\end{clm}

\ifpart
\begin{proof}
Consider $\hat{G}$ for any graph $G$ such that $C(G)$ is not complete.
Every vertex in $I(G)$ is still a singleton in $\hat{G}$. Thus, all of
the edges in $\hat{G}$ are confined to the remaining vertices:
$c(G)-s(G)$ of them. So $\theta(\hat{G})\leq\mmax{c(G)-s(G)}+i(G)$.
Moreover, since $C(G)$ is not complete, we know from
Lemma~\ref{lem:ninth} that $\theta(G)=\theta(\hat{G})$, so
$\theta(G)\leq\mmax{c(G)-s(G)}+i(G)$.
For any graph $G$ such that $C(G)$ is complete, $\theta(C(G))=1$, so
$\theta(G)=1+i(G)$.
\qed
\end{proof}
\fi

\begin{clm}\label{clm:thirteenth}
If $G$ has a triangle, let $G_0$ be the graph which results from
removing two of the edges in said triangle.
$\theta(G_0)\geq\theta(G)-1$.
\end{clm}

\ifpart
This claim doesn't require proof; we could remove the whole triangle
and still decrease $\theta$ by at most, as it cannot take more than 1
clique to cover a clique.
\fi

\begin{lem}\label{lem:fourteenth}
If $m\geq\mmax{n}$, $G\in\mathcal{G}_n(m)$, and $i(G)\neq0$ then
$\theta(G)\neq\Theta_n(m)$. 
\end{lem}

\ifpart
\begin{proof}
Let $m\geq\mmax{n}$, and let $G\in\mathcal{G}_n(m)$ such that $i(G)\neq0$. Then
$m$ edges must fit on $n-i(G)$ vertices, so $m\leq{n-i(G)\choose2}$.

{\bf Case 1:} If $m={n-i(G)\choose2}$, then $C(G)$ is complete, so
$\theta(G)=1+i(G)$. Let $\{u,v\}$ be an edge in $C(G)$, and let $s$ be a
singleton in $I(G)$. Remove the edge $\{u,v\}$, and replace it with the edge
$\{u,s\}$ to attain graph $G'$. Since $s$ is part of exactly 1 edge in $G'$,
this edge is a maximal clique, so it must be in any cover of $G'$. To cover the 
remainder of $C(G')$, exactly 2 cliques are necessary: one including $v$ and not
$u$, and the other including $u$ and not $v$. Thus,
$\theta(G')=3+i(G')=2+i(G)>\theta(G)$.

{\bf Case 2:} Alternately, if $\mmax{n}\leq m<{n-i(G)\choose2}$, then $C(G)$ is
not complete. Thus, there is a pair of vertices $u$ and $v$ in $C(G)$ such that
$\{u,v\}$ is not an edge in $G$. Since $m\geq\mmax{n}>\mmax{n-1}$, there is at
least 1 triangle in $G$. Let $\{x,y,z\}$ be a triangle in $G$. Remove the edges
$\{x,y\}$ and $\{y,z\}$ to obtain graph $G_0$. We know from
Lemma~\ref{clm:thirteenth} that $\theta(G_0)\geq\theta(G)-1$. Next, choose
some singleton $s$ in $G_0$, and add to $G_0$ the edges $\{u,s\}$ and $\{v,s\}$
to get graph $G'$. Since $s$ was isolated and $\{u,v\}$ was not an edge in
$G_0$, these two new edges are maximal cliques. Thus,
$\theta(G')=\theta(G_0)+2\geq\theta(G)+1$.

In both cases, we have found a graph $G'\in\mathcal{G}_n(m)$ such that
$\theta(G')>\theta(G)$. Thus, $\theta(G)\neq\Theta_n(m)$.
\qed
\end{proof}
\fi

\begin{lem}\label{lem:fifteenth}
For $G\in\mathcal{G}_n(m)$, if $1<m<{n\choose2}$ and $C(G)$ is complete then
$\theta(G)\neq\Theta_n(m)$. 
\end{lem}

\ifpart
\begin{proof}
{\bf Case 1:} Let $1<m\leq\mmax{n}$, and consider $G\in\mathcal{G}_n(m)$ such 
that $C(G)$ is complete. Since $m>1$, there are at least 3 vertices in $C(G)$, 
so $G$ contains a triangle. Thus, by Lemma~\ref{lem:fifth},
$\theta(G)\neq\Theta_n(m)$.

{\bf Case 2:} Let $\mmax{n}\leq m<{n\choose2}$, and consider
$G\in\mathcal{G}_n(m)$ such that $C(G)$ is complete. Then $c(G)<n$, because it
would take ${n\choose2}$ edges to make a complete graph on $n$ edges. As such,
$i(G)\neq0$, so by Lemma~\ref{lem:fourteenth}, $\theta(G)\neq\Theta_n(m)$.
\qed
\end{proof}
\fi

\begin{lem}\label{lem:sixteenth}
$\Theta_n(\mmax{n}+1)=\mmax{n-1}$
\end{lem}

\begin{proof}
It can be shown quickly through enumeration of all arrangements of
edges that this is true for 3 and 4 vertices. We will prove it for
larger integers through induction.

Assume the Lemma holds for some even $n\geq4$. Let
$G\in\mathcal{G}_{n+1}(\mmax{n+1}+1)$ such that $i(G)=0$. $G$ has
degree sum $D=2\mmax{n+1}+2$. Let $v$ be a vertex of minimum degree in
$G$. For the sake of contradiction, assume $deg(v)>n/2$. Then
$deg(v)\geq n/2+1$, so every vertex in $G$ has degree $\geq n/2+1$.
Therefore, $D\geq(n+1)(n/2+1)=\frac{n^2+3n+2}{2}$.  Moreover, since
$n\geq4$,
$D\geq\frac{n^2+2n+6}{2}=\frac{n^2+2n+1}{2}+5/2>2\mmax{n+1}+2=D$; that
is, $D>D$. The assumption that $deg(v)>n/2$ led to a contradiction, so
$deg(v)\leq n/2$. Remove $v$, all of its incident edges, and
$n/2-deg(v)$ additional edges to construct graph $G'$. Since $n$ is
even, $\mmax{n+1}-\mmax{n}=\mmax{n}-\mmax{n-1}=n/2$. $G'$ has $1$
fewer vertices and $n/2$ fewer edges than
$G\in\mathcal{G}_{n+1}(\mmax{n+1}+1)$, so
$G'\in\mathcal{G}_n(\mmax{n}+1)$. By the hypothesis,
$\theta(G')\leq\mmax{n-1}$.  We will now reconstruct $G$ from $G'$,
taking note of any changes to $\theta$.  As such, whenever we add an
edge, it is an edge which had previously been removed from $G$ in the
construction of $G'$. First, add $v$ and any 1 of its incident edges;
$v$ necessarily had at least 1 incident edge as $i(G)=0$. This
addition increases $\theta$ by exactly 1, as the newly added edge
comprises a maximal clique. There are now $n/2-1$ edges missing from
$G$. Add them all back, noting that Claim~\ref{clm:fourth} (or, more
accurately, the same reasoning used to prove it), guarantees
that each additional edges increases $\theta$ by at most 1. Thus,
$\theta(G)\leq\theta(G')+1+n/2-1\leq\mmax{n-1}+n/2=\mmax{n}$. 

Assume the Lemma holds for some odd $n\geq3$. Let
$G\in\mathcal{G}_{n+1}(\mmax{n+1}+1)$ such that $i(G)=0$. $G$ has
degree sum $D=2\mmax{n+1}+2$. As such, the average degree in $G$ is
$\frac{n+1}{2}+\frac{2}{n+1}$. Thus, the minimum degree of any vertex
in $G$ is at most $\frac{n+1}{2}$.  

{\bf Case 1:} If the minimum degree in $G$ is $\leq\frac{n-1}{2}$, then
let $v$ be a vertex of degree $\leq\frac{n-1}{2}$.
Remove $v$ and all
of its incident edges, along with any additional edges
necessary to remove $\frac{n-1}{2}$ total, to get graph $G'$.
Of course, $\theta(G)\leq\theta(G')+\frac{n-1}{2}$, as we can simply add
the $\frac{n-1}{2}$ edges to any cover of $G'$ to obtain a cover of $G$.
 Moreover, $G'\in\mathcal{G}_n(\mmax{n}+2)$. Since $\Theta_n(m)$ is
non-increasing for $m\geq\mmax{n}$,
$\theta(G')\leq\Theta_n(\mmax{n}+1)$. So, by the hypothesis,
$\theta(G')\leq\mmax{n-1}$. Thus,
$\theta(G)\leq\mmax{n-1}+\frac{n-1}{2}$; identically,
$\theta(G)\leq\mmax{n}$. 

{\bf Case 2:} If the minimum degree in $G$ is $\frac{n+1}{2}$, then
let $v$ be a vertex of degree $\frac{n+1}{2}$.

{\bf Subcase 1:} If $v$ is in a triangle, remove it an all of its
edges to get $G'$. $G'\in\mathcal{G}_n(\mmax{n}+1)$, so
$\theta(G')\leq\mmax{n-1}$ by the hypothesis. Since two of these edges
were in a triangle, their removal reduced the clique cover size by at
most 1. Removing the remaining
$\frac{n-3}{2}$ edges reduced $\theta$ by at most $\frac{n-3}{2}$.
Therefore,
$\theta(G)\leq\theta(G')+\frac{n-1}{2}\leq\mmax{n-1}+\frac{n-1}{2}=\mmax{n}$;
that is, $\theta(G)\leq\mmax{n}$.

{\bf Subcase 2:} If $v$ is not in a triangle, we need only find a 
vertex $v_0$, in a triangle, with degree $\frac{n+1}{2}$ for the 
previous subcase to apply. None of $v$'s
$\frac{n+1}{2}$ neighbors are adjacent to each other. Let $V$ be the
set of vertices adjacent to $v$, and let $V'$ be the set of vertices
in $G$ with are neither $v$ nor adjacent to $v$. Note that
$|V|=\frac{n+1}{2}$ and $|V'|=\frac{n-1}{2}$. The vertices in $V$ all
have degree of at least $\frac{n+1}{2}$, as this is the minimum degree
in $G$. As such, every vertex in $V$ is adjacent to at least
$\frac{n-1}{2}$ vertices other than $v$; none of the vertices in $V$
are adjacent to each other, and there are only $\frac{n-1}{2}$ vertices
(other than $v$) remaining. Thus, every vertex in $V$ is adjacent to
all $\frac{n-1}{2}$ vertices in $V'$. Let's count edges: there are
$\frac{n+1}{2}$ between $v$ and $V$ and another $\mmax{n}$ between $V$
and $V'$, and there are $\mmax{n+1}+1$ total edges; 1 edge is
unaccounted for. This edge must be between 2 vertices in $V'$, as it
cannot be in $V$, and every edge from $V$ to $V'$ is already counted.
So let $v_0$ be any vertex in $V$. $deg(v_0)=\frac{n+1}{2}$ and $v_0$
is necessarily in a triangle, so apply the previous subcase.

We have shown that
$\Theta_n(\mmax{n}+1)=\mmax{n-1}\implies\Theta_{n+1}(\mmax{n+1}+1)\leq\mmax{n}$.
It is easy to construct a graph $G\in\mathcal{G}_{n+1}(\mmax{n+1}+1)$
such that $\theta(G)=\mmax{n}$; simply add an extra edge to
$K_{\lfloor\frac{n}{2}\rfloor,\lceil\frac{n}{2}\rceil}$ 
between any two vertices in the larger partition
(if $n$ is odd) or in either partition (if $n$ is even).
Therefore, for all values $n$, $\Theta_n(\mmax{n}+1)=\mmax{n-1}$.
\qed
\end{proof}

Lemma~\ref{lem:sixteenth} implies that
$\Theta_n(m)=\mmax{n-1}$, if $\mmax{n}<m\leq{n\choose2}-\mmax{n-2}$, as $\Theta_n(m)$ is
non-increasing for $m\geq\mmax{n}$ and Theorem~\ref{thm:lovasz}
implies that $\Theta_n({n\choose2}-\mmax{t})=\mmax{t+1}$.
This, shows that the following three statements are equivalent. We conject
that they are true:

\begin{cnj}\label{cnj:seventeenth}

\begin{enumerate}[label=(\roman*)]
	\item\label{cnj17i} If $m\geq\mmax{n}$, 
		then $\Theta_{n+1}(m+n)=\Theta_n(m)$.
	\item\label{cnj17ii} If $m={n\choose2}-\mmax{t}+1>\mmax{n}$, then 
		$\Theta_n(m)=\mmax{t}$.
	\item\label{cnj17iii} If $m\geq\mmax{n}$, let $k={n\choose2}-m$, 
		and let $t$ be the largest natural number such that $\mmax{t}\leq k$. 
		Then $\Theta_n(m)=\mmax{t+1}$.

\end{enumerate}
\end{cnj}

Lemma~\ref{lem:sixteenth} shows that \ref{cnj17ii} is true when $t=n-1$.
We will show in Lemma~\ref{lem:nineteenth} that \ref{cnj17ii} also holds 
when $t=n-2$,
and in Lemma~\ref{lem:eighteenth} that \ref{cnj17i} is true when
${n\choose2}-m\leq n/2$.

\begin{lem}\label{lem:eighteenth}
Let $k={n\choose2}-m$. If $k\leq n/2$, then $\Theta_{n+1}(m+n)=\Theta_n(m)$.
\end{lem}

\begin{proof}
Let $G\in\mathcal{G}_{n+1}(m+n)$. ${n+1\choose2}-(m+n)={n\choose2}-m=k$, so $G$
is missing $k$ edges. $k\leq n/2$, so $k<\frac{n+1}{2}$. Every edge only has 2
endpoints, so there is at least 1 vertex $v$ which is not adjacent to any of
the missing edges. $v$ must be in $S(G)$, as it is adjacent to every other
vertex in $G$. Therefore, $\theta(G-\{v\})=\theta(G)$. $G-\{v\}$ has $n$
vertices and $m$ edges, so $\theta(G)\leq\Theta_n(m)$. 

We can easily construct a
graph $G\in\mathcal{G}_{n+1}(m+n)$ such that $\theta(G)=\Theta_n(m)$; simply add
a star to a graph $G_0\in\mathcal{G}_n(m)$ such that $\theta(G_0)=\Theta_n(m)$.
\qed
\end{proof}

\begin{lem}\label{lem:nineteenth}
$\Theta_n({n\choose2}-\mmax{n-2}+1)\leq\mmax{n-2}$
\end{lem}

\ifreplace
Lemmas~\ref{lem:sixteenth} ,~\ref{lem:eighteenth}, and~\ref{lem:nineteenth} 
are our most notable results, as they are proven improvements to
Theorem~\ref{thm:lovasz}. We regret that we must the omit proof
of Lemma~\ref{lem:nineteenth}, as it is nearly four pages long.
\fi

\ifpart
\begin{proof}
It should first be noted that ${n\choose2}-\mmax{n-2}+1 = \mmax{n+1}$, so this
Lemma can be reformulated as $\Theta_n(\mmax{n+1})\leq\mmax{n-2}$.
We have shown through exhaustive search of all graphs on $\leq8$ vertices
that this is true for all $n\leq8$.

{\bf Case 1:} Even $n$

Assume $n$ is even and $n\ge10$. Let $G\in\mathcal{G}_n(m)$, where
$m={n\choose2}-\mmax{n-2}+1$. Moreover, assume that the Lemma is true for all
$n_0<n$.

Since $n$ is even, ${n\choose2}-\mmax{n-2}+1 = \frac{n^2+2n}{4}$. Thus, the
degree sum of $G$ is $\frac{n^2+2n}{2}$, and as such the average degree in $G$
is $\frac{n}{2}+1$. Let $d=\min{\{\text{deg}(v)|v\in V\}}$. Based on the average
degree, $d\leq\frac{n}{2}+1$.

{\bf Subcase 1:} $d\leq\frac{n}{2}-1$

If $d\leq\frac{n}{2}-1$, let $v$ be a vertex with degree $\leq\frac{n}{2}-1$,
and let $G'$ be $G$ without $v$ or any of its edges. $G'$ has $n-1$ vertices
and $m-d$ edges. $m-d\geq m-\frac{n}{2}+1=\mmax{n}+1\geq\mmax{n}$.
Thus, $\theta(G')\leq\Theta_{n-1}(\mmax{n})\leq\mmax{n-3}$. The remaining edges
(those adjacent to $v$) can be covered by at most $\frac{n}{2}-1$ cliques. Thus,
$\theta(G)\leq\mmax{n-3}+\frac{n}{2}-1=\mmax{n-2}$.

{\bf Subcase 2:} $d=\frac{n}{2}$

If $d=\frac{n}{2}$. Let $v$ be a vertex of degree $d$. Let $\mathcal{N}_v$ denote $v$'s neighborhood,
and let $\mathcal{N}_v^c$ denote the set of vertices in $G$ which are neither
$v$ nor in $\mathcal{N}_v$.

Assume $v$ is not in a triangle. $|\mathcal{N}_v|=\frac{n}{2}$, so
$|\mathcal{N}_v^c|=\frac{n}{2}-1$ (as there are $n-1$ vertices other than $v$).
Since $v$ is not in a triangle, $\mathcal{N}_v$ is pairwise disjoint. Every
vertex must be adjacent to at least $\frac{n}{2}$ others, and there are only
$\frac{n}{2}$ vertices not in $\mathcal{N}_v$, so every vertex in
$\mathcal{N}_v$ must be adjacent to all of them. Moreover, $\mathcal{N}_v$ is
pairwise disjoint, no more edges can be added to any vertex in it, so the degree
of every vertex in $\mathcal{N}_v$ is $\frac{n}{2}$ exactly. Let's count edges:
there are $\frac{n}{2}$ between $v$ and $\mathcal{N}_v$, and another
$\frac{n}{2}(\frac{n}{2}-1)$ between $\mathcal{N}_v$ and $\mathcal{N}_v^c$;
that's $\mmax{n}$ total edges; $\frac{n}{2}$ edges are missing. These edges must
be in $\mathcal{N}_v^c$, as $v$ is not in any triangles. Since every vertex in
$\mathcal{N}_v$ is connected to every vertex in $\mathcal{N}_v^c$, a single edge
in $\mathcal{N}_v^c$ is enough to imply that there is an element of
$\mathcal{N}_v$ which is in a triangle. Since every element of $\mathcal{N}_v$
has degree $\frac{n}{2}$, there must be a vertex of degree $\frac{n}{2}$ which
is in a triangle.

So it is safe to assume that there is a vertex of degree $\frac{n}{2}$ which is in
a triangle, as such a vertex necessarily exists. Let $v$ be such a vertex. Let
$G'$ be $G$ without $v$ or its edges. $G'$ has $n-1$ vertices and $\mmax{n}$
edges, so $\theta(G')\leq\Theta_{n-1}(\mmax{n})$. This, by our hypothesis,
implies that $\theta(G')\leq\mmax{n-3}$. In order to cover the rest of $G$, we
need only cover the edges adjacent to $v$. There are $\frac{n}{2}$ of them, 2 of
which can be covered by 1 triangle. Thus,
$\theta(G)\leq\mmax{n-3}+\frac{n}{2}-1=\mmax{n-2}$.

{\bf Subcase 3:} $d=\frac{n}{2}+1$

Finally, if $d=\frac{n}{2}+1$, then every vertex has degree of exactly $d$.
Since $m>\mmax{n}$, there is a triangle $T=(u,v,w)$ in $G$. Let $G'$ be $G$
without $u$, $v$, $w$ or their edges. $G'$ has $n-3$ vertices and $\mmax{n-2}-1$
edges, so by Lemma~\ref{lem:sixteenth}, $\theta(G')\leq\mmax{n-4}$. Each vertex
in $T$ has $\frac{n}{2}-1$ neighbors in $G'$ (as its other 2 neighbors are in
$T$. This gives us a lot of information. First, since there are $n-3$ vertices
in $G'$ and $2(\frac{n}{2}-1)=n-2$, each pair of vertices in $T$ has at least 1
common neighbor in $G'$. This means that all three edges in $G'$ can be covered
by triangles containing 2 elements of $T$ and 1 element of $G'$. Moreover,
$3(\frac{n}{2}-1)=\frac{n}{2}+n-3$; if we were to list all of the neighbors (in
$G'$) of the vertices in $T$, there would be at least $\frac{n}{2}$ repeats (if
a vertex is listed twice, it's been repeated once, and if it's been listed three
times it's been repeated twice\dots). Let $x$ be a repeated vertex. If $x$ was
only repeated once, it is adjacent to two elements of $T$. As such, two edges
connecting $T$ to $G'$ can be covered with 1 triangle. As such, this repeat
allows us to cover at least 1 ``extra'' edge with 1 clique. Alternately, if a
vertex is repeated twice, it is adjacent to all 3 elements of $T$, so 3 edges
between $T$ and $G'$ can be covered with 1 clique, so we've covered 2 ``extra''
edges. Moreover, these triangles and 4-cliques do not repeat edges (other than
those in $T$, as they each have a single unique vertex in $G'$. So these
$\frac{n}{2}$ repetitions equate to at least $\frac{n}{2}$ fewer cliques needed
to cover the edges between $T$ and $G'$, while simultaneously covering all 3
edges in $T$. As such, we can cover all edges in $T$ and all edges between $T$
and $G'$ with at most $3(\frac{n}{2}-1)-\frac{n}{2}=n-3$. 

So $\theta(G)\leq\theta(G')+n-3\leq\mmax{n-4}+n-3=\mmax{n-2}$.

{\bf Case 2:} Odd $n$

Assume $n$ is odd and $\geq9$. Let $G\in\mathcal{G}_n(m)$, where $m=\mmax{n+1}$.
Moreover, like the previous case, assume that the Lemma is true for all
$n_0<n$. 

$m=\mmax{n+1}=\frac{n^2+2n+1}{4}$, so the degree sum is $\frac{n^2+2n+1}{2}$. As
such, the average degree is $\frac{n+1}{2}+\frac{n+1}{2n}$, so
$d\leq\frac{n+1}{2}$ (as $\frac{n+1}{2}$ is obviously less than 1).

{\bf Subcase 1:} $d\leq\frac{n-3}{2}$

If $d\leq\frac{n-3}{2}$, let $v$ be a vertex with degree $d$. Let $G'$ be $G$
without $v$ or its edges. $G'$ has $n-1$ vertices and at least $\mmax{n}+2$
edges. Thus, $\theta(G')\leq\mmax{n-3}$. We can cover the edges adjacent to $v$
with at most $\frac{n-3}{2}$ cliques, so
$\theta(G)\leq\mmax{n-3}+\frac{n-3}{2}=\frac{n^2-4n+3}{4}=\mmax{n-2}$. 

{\bf Subcase 2:} $d=\frac{n-1}{2}$

Let $v$ be a vertex of degree $d$. There are $\frac{n-1}{2}$ vertices in
$\mathcal{N}_v$.

Assume $v$ is not in a triangle. Every vertex in $\mathcal{N}_v$ is adjacent to
at least $\geq\frac{n-3}{2}$ vertices and $\leq\frac{n-1}{2}$ in 
$\mathcal{N}_v^c$. Thus, there are at most $(\frac{n-1}{2})^2=\mmax{n-1}$ edges
between $\mathcal{N}_v$ and $\mathcal{N}_v^c$. We have accounted for
$\mmax{n-1}+\frac{n-1}{2}=\mmax{n}$ edges; there are at least $\frac{n-1}{2}$
edges remaining, all of which must be in $\mathcal{N}_v^c$.

If there is a vertex$v_0$ in $\mathcal{N}_v$ with degree $\frac{n-1}{2}$, this 
vertex is adjacent to $\frac{n-3}{2}$ of the $\frac{n-1}{2}$ vertices in
$\mathcal{N}_v^c$; let the vertex in $\mathcal{N}_v^c$ to which it is
nonadjacent be $w$. Given that there are at least $\frac{n-1}{2}$ edges in 
$\mathcal{N}_v^c$, which is strictly more than $\frac{n-3}{2}$, at least one of
the edges in $\mathcal{N}_v^c$ does not include $w$, and therefore forms a
triangle with its two adjacent vertices ($u_1$ and $u_2$) and $v_0$. We have
found $v_0$, a vertex with degree $\frac{n-1}{2}$ which is also in a triangle.
Let $G'$ be $G$ without $v_0$ or its edges. $G'$ has $n-1$ vertices and least
$\mmax{n}+1$ edges, so $\theta(G')\leq\Theta_{n-1}(\mmax{n}+1)\leq\mmax{n-3}$.
The remaining edges in $G$ are those adjacent to $v_0$; there are
$\frac{n-1}{2}$ of them, and 2 can be covered by a triangle, so
$\theta(G)\leq\mmax{n-3}+\frac{n-3}{2}=\mmax{n-2}$. Also note that this works
with any vertex of degree $\frac{n-1}{2}$ which is in a triangle, so the case in
which $v$ is in a triangle has been proven in the process.

If there are no vertices of degree $\frac{n-1}{2}$ in $\mathcal{N}_v$, then
every vertex in $\mathcal{N}_v$ has degree $\frac{n+1}{2}$ (as there are only
$\frac{n+1}{2}$ vertices not in $\mathcal{N}_v$). Therefore, every vertex in
$\mathcal{N}_v$ is adjacent to every vertex in $\mathcal{N}_v^c$. So there are
$(\frac{n-1}{2})^2=\mmax{n-1}$ edges between $\mathcal{N}_v$ and
$\mathcal{N}_v^c$, and $\frac{n-1}{2}$ between $v$ and $\mathcal{N}_v$ for a
total of $\mmax{n}$; there are $\frac{n}{2}$ edges in $\mathcal{N}_v^c$. 

Let $E$ be the set of edges in $\mathcal{N}_v^c$. Assume, for contradiction,
that none of the elements of $E$ are pairwise disjoint. That is, assume every 
pair of edges in $E$ has a common vertex. Then $\mathcal{N}_v^c$ contains no 
triangles, as any edge in $E$ (not in the triangle) could contain at most 1 
of the triangles vertices and would therefore be disjoint with at least 1 of the
triangles edges. 

Choose any $2$ edges in $E$. These edges share a
vertex, $a$. Choose a third edge in $E$. If it is not adjacent to $a$, then it
must form a triangle with the previous 2 edges in order to share a vertex with
each of them. Therefore, the third edge must be adjacent to $a$. As must the
fourth\dots until there are $\frac{n-3}{2}$ edges connecting $a$ to every
other vertex in $\mathcal{N}_v^c$. There must be another edge in
$\mathcal{N}_v^c$, and it cannot be adjacent to $a$; label this edge $(c,d)$. 
Since $n\geq9$, there are $\geq4$ vertices in $\mathcal{N}_v^c$, so there is at
least $1$ more vertex, $b$. Since $a$ is adjacent to everything in
$\mathcal{N}_v$, $(a,b)$ is also an edge. We have found a pair of disjoint edges
in $E$, so obviously our assumption that it is not pairwise disjoint was false;
there are two disjoint edges $(a,b)$ and $(c,d)$ in $\mathcal{N}_v^c$.

Let $v_0$ be a vertex in $\mathcal{N}_v$. Let $G'$ be $G$ without $v_0$ or its
edges. $G'$ has $n-1$ vertices and $\mmax{n+1}-\frac{n+1}{2}=\mmax{n}$. Thus,
$\theta(G)\leq\Theta_{n-1}(\mmax{n})\leq\mmax{n-3}$. We can cover the remaining
$\frac{n+1}{2}$ edges adjacent to $v_0$ with at most $\frac{n-3}{2}$ cliques,
because $4$ of these edges can be covered by the triangles $(a,b,v)$ and
$(c,d,v)$. Thus, $\theta(G)\leq\mmax{n-3}+\frac{n-3}{2}=\mmax{n-2}$. Note that
this works with any vertex of degree $\frac{n+1}{2}$ which is in two otherwise
disjoint triangles.

{\bf Subcase 3} $d=\frac{n+1}{2}$

Let $V$ be the set of vertices in $G$ with degree $\frac{n+1}{2}$, and let $V'$
be the vertices with degree $\geq\frac{n+3}{2}$. Assume, for contradiction, that
$|V|\leq\frac{n-3}{2}$. Then the minimum degree sum of $G$ is
$(\frac{n-3}{2})(\frac{n+1}{2})+(\frac{n+3}{2})(\frac{n+3}{2})=\frac{n^2+2n+3}{2}$.
Recall that the degree sum of $G$ is $\frac{n^2+2n+1}{2}$, so we've found our
contradiction; we now know that $|V|\geq\frac{n-1}{2}$.

Assume, for contradiction, that $|V'|\leq1$. Then the maximum degree sum of $G$ is
$(n-1)(\frac{n+1}{2})+1(\frac{n+3}{2})=\frac{n^2+n+2}{2}<\frac{n^2+2n+1}{2}$
(recall $n\geq4$). This contradicts the known degree sum of $G$, so it must be 
false; $|V'|\geq2$. 

Once more, assume for contradiction that there is no element of $V$ which is
adjacent to two distinct elements of $V'$. Every vertex in $V$ has
$\frac{n+1}{2}$ neighbors. At most one of these neighbors is in $V'$, so at least
$\frac{n-1}{2}$ neighbors must be in $V$. Therefore $|V|\geq\frac{n+1}{2}$.
As such, each there are at most $\frac{n-1}{2}$ vertices in $V'$, so each vertex
in $V'$ must be adjacent to at least four vertices in $V$. Moreover, since no
vertex in $V$ can be adjacent to two in $V'$, we now know that $|V|\geq4|V'|$;
in other words, $|V'|\leq\lfloor\frac{n}{5}\rfloor$. 

Choose two vertices in $V'$. They must each be adjacent to at least 
$\frac{n+1}{2}$ vertices (other than each other). There are only $n-2$ other 
vertices, so they have at least 3 neighbors in common. These common neighbors
are adjacent to two vertices in $V'$, so they cannot be in $V$ and must be in
$V'$; $|V'|\geq5$. Choose five vertices in $V'$. They are each adjacent to at
least $\frac{n-3}{2}$ vertices (again, other than each other), of which there are
$n-5$. $5(\frac{n-3}{2})-(n-5)=\frac{3n-5}{2}$, so there are at least
$\frac{3n-5}{2}$ repeats. A vertex can be repeated at most four times (five
vertices in $V'$, first isn't a repeat), so there are at least $\frac{3n-5}{8}$
repeated vertices; that is, there are at least $\frac{3n-5}{8}$ vertices (other
than the give we already had) in adjacent to two elements of $V'$, and therefore
in $V'$. Thus, $|V'|\geq\frac{3n-5}{8}+5$.

We know that $\frac{3n-5}{8}+5\leq|V'|\leq\frac{n}{5}$, so if
$\frac{3n-5}{8}+5>\frac{n}{5}$, we have a contradiction. This turns out to be
true for any value of $n$ which is $\geq-25$, which $n$ obviously is, so at long
last we've arrived at a contradiction. Therefore, there must be an element $v$
of $V$ which is adjacent to two elements $a$ and $b$ of $V'$. 

Assume $a$ and $b$ are adjacent. $v$ is adjacent to $\frac{n-3}{2}$ vertices
other than $a$ and $b$. $a$ is adjacent to at least $\frac{n-1}{2}$ vertices
other than $b$ and $v$. There are $n-3$ total other vertices, so $a$ and $v$
have at least one neighbor $a_1\neq b$ in common. Similarly, $b$ and $v$ have a
neighbor $b_1\neq a$ in common. If $b_1=a_1$, then $a_1$ is adjacent to $b$, so 
$(v,a,b,a_1)$ is a clique. This clique can cover three edges between $v$ and
$\mathcal{N}_v$. If $a_1\neq b_1$, then the two triangles $(v,a,a_1)$ and
$(v,b,b_1)$ cover four edges between $v$ and $\mathcal{N}_v$. In either case,
every edge between $v$ and its neighborhood can be covered with at most
$\frac{n-3}{2}$ cliques. 

If $a$ and $b$ are not adjacent, then $a$ is adjacent to at least
$\frac{n+1}{2}$ vertices other than $v$. $v$ is still adjacent to
$\frac{n-3}{2}$ vertices other than $a$ and $b$, so $a$ and $v$ have at least
two neighbors $a_1,a_2$ in common (and since $a$ and $b$ are not adjacent,
these common neighbors cannot be $b$). Similarly, $b$'s and $v$'s neighborhoods
have at least $b_1,b_2$ in common. At least one of $\{b_1,b_2\}$ is not $a_1$,
so we can assume without loss of generality that $a_1\neq b_1$. The
two triangles $(v,a,a_1)$ and $(v,b,b_1)$ cover four edges between $v$ and
$\mathcal{N}_v$, so the edges adjacent to $v$ can be covered with at most
$\frac{n-3}{2}$ cliques.

So, whether or not $a$ and $b$ are adjacent, let $G'$ be $G$ without $v$ or its
edges. $G'$ has $n-1$ vertices and $\mmax{n}$ edges, so
$\theta(G)\leq\Theta_{n-1}(\mmax{n})\leq\mmax{n-3}$. We've shown that the
remaining edges can be covered with at most $\frac{n-3}{2}$ cliques, so
$\theta(G)\leq\mmax{n-3}+\frac{n-3}{2}=\mmax{n-2}$.
\qed
\end{proof}
\fi

We haven't proven that Conjecture~\ref{cnj:seventeenth} is true for all $m$.
The results of Lemma~\ref{lem:sixteenth}
and~\ref{lem:nineteenth} and the non-decreasing nature of $\Theta_n(m)$ 
past $\mmax{n}$ show that Conjecture~\ref{cnj:seventeenth}\ref{cnj17iii} 
is true for all $m\in[\mmax{n},{n\choose2}-\mmax{n-3}]$. 
Similarly, Lemma~\ref{lem:eighteenth} shows that 
Conjecture~\ref{cnj:seventeenth}\ref{cnj17i} is true if 
$m\geq{n\choose2}-n/2$.


\section{Summary of results}

Putting together the results in Section~\ref{sec:pre-maximum}, the
behavior of $\Theta_n(m)$ for $m\le\mmax{n}$, 
and Section~\ref{sec:post-maximum}, the
behavior of $\Theta_n(m)$ for $m\ge\mmax{n}$, we can now state the
conclusion of the paper succinctly as Theorem~\ref{thm:final}.

\begin{thm}\label{thm:final}
Let $\Theta_n(m)$, for $1\le m\le{n\choose 2}$, be the size of the largest
minimal clique-cover for any graph on $n$ vertices and $m$ edges. Then, for
$m\le\mmax{n}$, we have:
\begin{subnumcases}{\Theta_n(m)=}
\Theta_{n-1}(m)+1 & \text{for $m\le\mmax{n-1}$} \label{thm:con1}\\
m & \text{for $\mmax{n-1}+1\le m\le\mmax{n}$} \label{thm:con2}\\
\mmax{n-1} & \text{for $\mmax{n}+1\le m\le{n\choose2}-\mmax{n-2}$} 
\label{thm:con3}\\
\mmax{n-2} & \text{for ${n\choose2}-\mmax{n-2}+1\le m\le{n\choose2}-\mmax{n-3}$} 
\label{thm:con4}\\
k+t & \text{for $m\geq{n\choose2}-\mmax{n-3}+1$}
\label{thm:con5}
\end{subnumcases}
Where $k$ and $t$ are defined as in Theorem~\ref{thm:lovasz}.

\noindent The recursive definition can be easily unwound, and the values of
$\Theta_n(m)$ computed explicitly with Algorithm~\ref{alg:ccb} which
runs in linear time in the length of the binary encoding of $n$.
\end{thm}

\begin{proof}
(\ref{thm:con1}) follows from Lemma~\ref{lem:eighth}.
(\ref{thm:con2}) follows from several sources: the left bound, i.e.,
$m=\mmax{n-1}+1$ is Lemma~\ref{lem:seventh}, the right bound, i.e., $m=\mmax{n}$
is Mantel (Theorem~\ref{thm:mantel}), and the value, which consists
of a subset of the line $y=x$, follows from Claim~\ref{clm:fourth} which limits
the growth, and therefore imposes an increase of 1 at each step.
(\ref{thm:con3}) and (\ref{thm:con4}) follow from Lemmas~\ref{lem:sixteenth}
and~\ref{lem:nineteenth} combined with Lov{\'a}sz (Theorem~\ref{thm:lovasz}).
(\ref{thm:con5}) is the bound provided by Lov{\'a}sz.
\qed
\end{proof}

Given proof of Conjecture~\ref{cnj:seventeenth} for all $m$, we can improve
Theorem~\ref{thm:final} to:

\begin{subnumcases}{\Theta_n(m)=}
\Theta_{n-1}(m)+1 & \text{for $m\le\mmax{n-1}$} \label{thm2:con1}\\
m & \text{for $\mmax{n-1}+1\le m\le\mmax{n}$} \label{thm2:con2}\\
\mmax{n-1} & \text{for $\mmax{n}+1\le m\le\mmax{n+1}-1$} \label{thm2:con3}\\
\Theta_{n-1}(m-(n-1)) & \text{for $m\ge\mmax{n+1}$} \label{thm2:con4}
\end{subnumcases}

Algorithm \ref{alg:ccb} (below), reflects the second version 
of Theorem~\ref{thm:final} (i.e. the one directly above), 
because this version grants an 
exact $\Theta_n(m)$, as opposed to an upper bound. 
It could be adjusted to reflect the original
theorem (roughly, omitting the improvements made in Lemmas~\ref{lem:sixteenth}
and~\ref{lem:nineteenth}) by replacing lines 18-21 with the single line
``return $k+p$''. 

\begin{algorithm}[H]
	\caption{$\Theta_n(m)$}\label{alg:ccb}
	\begin{algorithmic}[1]
	\Require Integers $n,m$ such that $n>0$ and $0\leq m\leq{n\choose2}$.
		\State $max \leftarrow \mmax{n}$
		\If{$m=max$}
			\State\Return $m$
		\EndIf
		\If{$m<max$}
			\State $p \leftarrow \sqrt{m}$
			\If{$p=\lfloor p\rfloor$}
				\State\Return $n+p(p-2)$
			\EndIf
			\State $p \leftarrow \lfloor p\rfloor$
			\If{$m\leq p(p+1)$}
				\State\Return $m+n-2p-1$
			\EndIf
			\State\Return $m+n-2p-2$
		\EndIf
		\State $k \leftarrow {n\choose2}-m$
		\State $p \leftarrow \lfloor\frac{1+\sqrt{1+4k}}{2}\rfloor$
		\If{$k<p^2$}
			\State\Return $p^2$
		\EndIf
		\State\Return $p(p+1)$
    \end{algorithmic}
  \end{algorithm}


\end{document}